\documentclass[5p, times]{elsarticle}

\usepackage{ecrc}

\volume{00}

\firstpage{1}

\journalname{Systems \& Control Letters}

\runauth{}
\jid{procs}
\jnltitlelogo{Systems \& Control Letters}
\CopyrightLine{2024}{Published by Elsevier Ltd.}

\usepackage{amssymb}
\usepackage{amsmath}
\usepackage{amsfonts}
\usepackage{xcolor}
\usepackage{amsthm}
\usepackage{mathtools}
\usepackage{appendix}
\usepackage{bm}
\usepackage{float}
\usepackage{hyperref}
\usepackage{tikz}
\usetikzlibrary{shapes.geometric, arrows.meta, positioning}

\tikzset{
  block/.style={
    draw, rectangle,
    minimum height=1.1cm,
    inner xsep=14pt,
    inner ysep=8pt,
    align=center,
    font=\normalsize
  },
  arrow/.style={-{Stealth[length=6pt]}, thick}
}

\newtheorem{theorem}{Theorem}

\newtheorem{example}{Example}
\newtheorem{corollary}{Corollary}
\newtheorem{lemma}{Lemma}
\newtheorem{proposition}{Proposition}

\usepackage[figuresright]{rotating}

\begin{document}

\begin{frontmatter}


\dochead{}

\title{Families of Control-Cost-Parametrized\\ Inverse-Optimal Universal Stabilizers}

\author[maeDept]{Miroslav Krsti\'c}
\author[eceDept]{Luke Bhan\corref{corresponding}
\cortext[corresponding]{Corresponding Author}
\tnotetext[funding]{The work of M. Krstic was funded by AFOSR grant FA9550-23-1-0535 and NSF grant ECCS-2151525. The work of L. Bhan is supported by
the U.S. Department of Energy (DOE) grant DE-SC0024386.}}

\address[maeDept]{Department of Mechanical and Aerospace Engineering, University of California San Diego, La Jolla, CA 92093-0411, USA}
\address[eceDept]{Department of Electrical and Computer Engineering, University of California San Diego, La Jolla, CA 92093-0411, USA}

\begin{abstract}
A classical universal stabilization formula offers the practitioner no design freedom: it is a single, parameter-free object. We introduce a cost-parametrized family of stabilizing feedback laws, where (1) the user chooses a function that serves as the running cost on control in an inverse-optimal cost functional, and (2) obtains, through a formula, a nonlinear ``expander'' of a pre-existing universal controller, which solves an infinite-horizon optimal control problem with a meaningful cost on the state. The cost-to-expander formula is a three-step construction, involving, inter alia,  cost differentiation and function inversion---overall, a nonlinear infinite-dimensional operator. The cost-to-expander operator is proven Lipschitz, which enables uniform neural operator approximation of the entire family and supports both offline performance exploration and online adaptation. Semiglobal practical asymptotic stability and second-order suboptimality bounds are established under the approximation. The operator learning and its use in semiglobal stabilization are illustrated numerically. We call the result `half-direct-optimal' because the paper's design is less than a general `direct optimal' (HJB-inducing) control, but more than the fully inverse optimal, since the user performs minimization for an arbitrary given cost on control. The dual to the half-direct problem we solve is the problem in which the cost on the state is arbitrary and given. This dual problem is easier and outside of the scope of the paper. 
\end{abstract}

\end{frontmatter}

\section{Introduction}

\paragraph{Background}
Sontag's formula \cite{SONTAG1989117} is a landmark result in nonlinear control: given any system admitting a control Lyapunov function, it produces an explicit, globally stabilizing feedback law. In this sense it is universal. Yet from the standpoint of control practice, this universality is of limited utility. The formula is a single object. It has no design parameters. The parameters that might exist within the CLF $V(x)$ are a part of a design process separate from the formula itself. The engineer who wishes to shape transient behavior, and is not a master of CLF (re-)design, has not a single knob to turn. The result is theoretically elegant but the practitioner is justified to be reluctant to take it from the theoretical realm into application.

On the opposite end of the spectrum from Sontag's formula is the regrettably insufficiently-known Curtis--Beard construction \cite{1310464,https://doi.org/10.1002/rnc.961}, which captures {\em all} inverse-optimally stabilizing controllers associated with a given CLF, though it does not offer a systematic approach for how to choose among the controllers. The Curtis--Beard formula can be viewed as a  generalization of Sontag's baseline. For the setting of our paper, single-input and disturbance-free, $\dot x=f(x)+g(x)u$, we rewrite the Curtis--Beard parametrization as
\begin{eqnarray}
u_{\rm CB}(x)=u_{\rm S}(x)
-\hat \pi(x)\,b(x) = -\pi(x) \, b(x)
\end{eqnarray}
where  $a=L_fV$, $b=L_gV$, $\pi_0=\dfrac{a+\sqrt{a^2+b^4}}{b^2}$, $u_{\rm S}= -\pi_0\,b$ denotes Sontag's formula, 
and the scalar gain adjustment $\hat \pi(x)$  
is chosen to satisfy
\begin{equation}
\hat \pi(x)\ge
\frac{\rho(x)-b(x)^2}
{\sqrt{a(x)^2+\rho(x)\,b(x)^2}+\sqrt{a(x)^2+b(x)^4}}\,,
\end{equation}
or, equivalently,
\begin{equation}
\pi \geq \frac{a+\sqrt{a^2+\rho \, b^2}}{b^2}\,.
\end{equation}
Thus, relative to the Sontag baseline $u_{\rm S}$, the designer is given {\em two} additional {\em functional} degrees of freedom
$(\rho, \hat \pi)$,
both with {\em vector} arguments and constrained.  
While this characterization is complete, it yields a design problem that is difficult to interpret and tune in practice: the effect of these functions on closed-loop transients and control effort is indirect and not readily predictable. A selection  requires extensive simulation across initial conditions and combinations of the design functions.
\paragraph{This paper's aim and results} The paper proposes a remedy to both the lack of design freedom's in Sontag's formula and the hard-to-use overabundant freedom of the Curtis--Beard formula.

We introduce a parametrization of a family of stabilizing feedback laws, indexed neither by a finite set of scalar parameters nor by two state-dependent functions as in Curtis--Beard, but by a single scalar function $\gamma$, which serves as the inverse-optimal cost  on the control input. Each admissible choice of $\gamma$ produces, through a three-step construction --- differentiation, an algebraic contraction, and function inversion --- a nonlinear expander $\kappa$, from which a stabilizing feedback law is assembled. The classical Sontag formula emerges as one special case, corresponding to $\kappa(s) = 2s$. The half-Sontag formula, stabilizing but not inverse-optimal as the full-Sontag formula, provides a natural nominal law that can be reshaped by varying $\gamma$. The practitioner's workflow is then concrete: select a family of functions $\gamma$, each encoding a different preference over control cost, generate the corresponding family of expanders $\kappa$, and choose among the resulting controllers on the basis of closed-loop performance.

The framework is, deliberately, both theoretical and practical in its ambitions. On the theoretical side, we establish that the map $\gamma \mapsto\kappa$  is Lipschitz on compact subsets of the relevant function space, and that the resulting family of controllers can be uniformly approximated by neural operators. On the practical side, this approximation result is not merely an existence statement: it means that a single trained neural operator can serve as a  surrogate for the entire parametric family. In real time, the neural operator can enable online controller selection and adaptation. Offline, the neural operator can serve for massive offline simulation testing of performance of feedback laws within the family, and enable the selection of those controllers that provide an adequate tradeoff between performance and control effort. The goal, in short, is to transform universal control: from an explicit embodiment of an existence result --- into an engineering tool, whose usage "at scale" is facilitated by machine learning.

A further property of the family deserves emphasis. All members of the universal formula family are endowed with not just stabilization but also with the property of inverse optimality, not observed as an afterthought, but built into the parametrization from the outset. Every controller in the family minimizes an infinite-horizon cost functional whose control penalty is precisely the function $\gamma$ the designer selects, with the knowledge of how aggressive or cautious he/she wishes the controller to be. 
In classical inverse-optimal design, one typically constructs a controller first and then identifies, if possible, a cost functional it minimizes; here the order is reversed, and the control cost comes first. The expander $\kappa$ and the feedback law are consequences. The result is a family in which performance shaping and optimality are not in tension but are two readings of the same object.

We refer to this formulation as ``half-direct-optimal'' since the (possibly non-quadratic) cost function on the control is general and given upfront, whereas the cost on the state is not general and, while meaningful, is the consequence of the control design and of the arbitrary cost on control. 

\paragraph{Related literature.}

The study of constructing controllers from CLFs is both vast and rich, beginning with the aforementioned Sontag's formula \cite{SONTAG1989117}, itself an explicit universal construction of Artstein's earlier stabilizability result \cite{ARTSTEIN19831163}. From here, there are too many extensions to mention, but we briefly highlight a few, including universal formulas for bounded and restricted inputs \cite{LIN1991393,LinSontag1995}, vectorized CLF's \cite{6519314}, integral-input-to-state stable (iISS) CLFs \cite{LIBERZON2002111}, and the nonsmooth/discontinuous feedback stabilization line \cite{633828}. More recently, there have been a series of results on safe stabilization \cite{9030225} and a recent general formula for an arbitrary number of affine control constraints is given in \cite{mestres2025neuralnetworkbaseduniversalformulas}.

Additionally, along a similar thread of thinking is the analysis of stabilizing controllers as inverse-optimal ones. Inverse optimal designs have been long studied dating back to the 1960s \cite{10.1115/1.3653115,1100365} with a major advancement in the robust design by Freeman-Kokotovi\'c \cite{doi:10.1137/S0363012993258732}. Since inverse-optimal design has been extended with local convergence rates \cite{5491091} as well as multiple different control paradigms such as input-to-state stability \cite{661589,LU2026112770}, discrete time systems \cite{ORNELASTELLEZ201438}, stochastic systems \cite{DENG1997151,DO201541}, adaptive control \cite{LI19971459}, reinforcement learning \cite{10.5555/3104322.3104366}, delayed systems \cite{CAI2019549}, PDE systems \cite{KARAFYLLIS2026106293}, safety filters \cite{10130642} and more recently, for strictly positive systems \cite{krstic2026contractorexpanderuniversalinverseoptimal}. The methodology has also been applied across several applications including aircraft tracking \cite{1532394,DO2021100132} and power networks \cite{JOUINI20235451}. 

A third strand of literature relevant to this work is operator learning, which enters here through the approximation of the expander mapping $\gamma \mapsto \kappa$. Neural operators, introduced in $1995$ in \cite{chen1995Universal} and later developed with modern neural network architectures in \cite{lu2021Learning, li2021Fourier}, provide a natural framework for approximating infinite-dimensional maps. Consequently, they have recently been used to replace costly implicit computations in feedback laws across several control settings, including PDE backstepping \cite{10374221, WANG2025112351, ZHANG2026112809}, adaptive control \cite{lamarque2025Adaptive}, delayed systems \cite{pmlr-v283-bhan25a} optimal control \cite{sewell2026neuraloperatorsmultitaskcontrol,li2026fnoanglethetaextendedfourier}, and applications such as traffic flow and oil drilling \cite{pmlr-v242-zhang24c, TOUMI2025106191}. 

\paragraph{Paper's specific contributions}
\begin{itemize}
\item \emph{A half-direct-optimal design paradigm for universal stabilization.}
We introduce a family of stabilizing feedback laws in which the practitioner prescribes a single scalar function $\gamma$ specifying the running cost on the control input. The constructed stabilizers are then optimal over an infinite-horizon  with control penalty $\gamma$ and an induced state cost $q(x)$, thereby allowing the aggressiveness of the controller to be shaped \emph{a priori} through the choice of $\gamma$.


\item \emph{Lipschitzness and learnability of the $\gamma \mapsto \kappa$ map.}
We prove that the infinite-dimensional operator mapping from control cost $\gamma$ to the corresponding expander $\kappa$ is Lipschitz on compact subsets of the admissible function space. Further, we invoke this regularity to guarantee arbitrarily accurate neural operator approximations of the expander operator which are needed for the practitioner's implementation. 

\item \emph{Stability and quantitative near-optimality under operator approximation.}
When the expander operator is uniformly approximated, we show that the resulting learned feedback law preserves semiglobal practical asymptotic stability. We further quantify the resulting suboptimality through a second-order distortion bound on the induced state-cost term and a finite-horizon optimality bound measuring the performance gap between the closed-loop trajectories of the approximate and exact half-direct-optimal controllers.\end{itemize}

\paragraph{Organization of the paper.}
The remainder of the paper is organized as follows. Section \ref{sec:contractor-expander-fcns} introduces the extractor and expander function along with their properties needed for the parametrization design. Section \ref{sec:half-direct-optimal-design} presents the corresponding half-direct-optimal stabilizing feedback framework. Section \ref{sec:lipschitz-approx} proves that the mapping from cost into expander function is Lipschitz and establishes a neural operator approximation theorem of this mapping over compact sets. Section \ref{sec:spas} contains the main robustness result when using the learned expander function in the parameterized stabilizing control law. Section \ref{sec:proofs} contains the proofs of the main theorem and Section \ref{sec:num} presents a numerical illustration of the design. 


\section{Contractor and Expander Functions}\label{sec:contractor-expander-fcns}

\begin{figure}
    \centering
\begin{tikzpicture}[node distance=.7cm]

  \node (input) {$\gamma$};

  \node (diff) [block, below=of input] {Differentiation \\ $\gamma'$};

  \node (alg)  [block, below=of diff]  {%
    \begin{tabular}{c}
      Algebraic mapping \\[4pt]
      $\displaystyle \Theta(s) = s - \frac{\gamma(s)}{\gamma'(s)}$
    \end{tabular}
  };

  \node (inv)  [block, below=of alg]  {Inversion \\ $\kappa = \Theta^{-1}$};

  \node (output) [below=of inv] {$\kappa$};

  \draw[arrow] (input) -- (diff);
  \draw[arrow] (diff)  -- (alg);
  \draw[arrow] (alg)   -- (inv);
  \draw[arrow] (inv)   -- (output);

\end{tikzpicture}
   
   \caption{Construction of the nonlinear {\em expander operator} $\mathcal K: \gamma \mapsto\kappa$  via differentiation of $\gamma$, algebraic operation $\Theta(s)=s-\gamma(s)/\gamma'(s)$, and inversion of function.}
    \label{fig:gamma-to-kappa}
\end{figure}
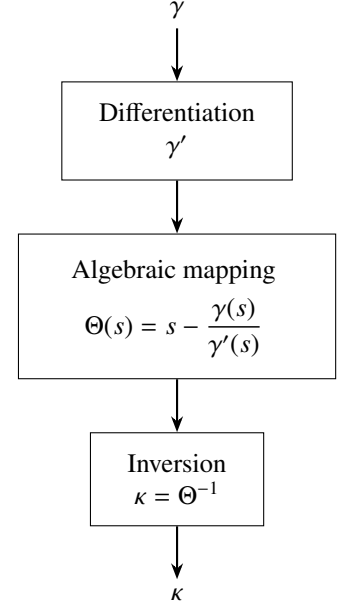

The following lemma introduces the contractor and expander functions and establishes their
basic properties, chief among them that $\Theta < \mathrm{Id}$ and $\kappa > \mathrm{Id}$
on $(0,\infty)$---which justify the nomenclature. The three-step construction is illustrated
in Figure~\ref{fig:gamma-to-kappa}.

\begin{lemma}\label{lem:Theta_kappa_Kinf}
Let $\gamma:[0,\infty)\to[0,\infty)$ be continuous on $[0,\infty)$, twice continuously differentiable on $(0,\infty)$, and satisfy $\gamma(0)=0$. Assume that $\gamma'\in \mathcal K_\infty$ and that
\begin{equation}\label{eq:gamma_dd_pos}
\gamma''(s)>0 \qquad \text{for all } s>0.
\end{equation}
Define the {\em contractor} function as
\begin{equation}\label{eq:Theta}
\Theta(s):=s-\frac{\gamma(s)}{\gamma'(s)}, \qquad s>0,
\end{equation}
and set $\Theta(0):=0$. Let the {\em expander} function be given by
\begin{equation}\label{eq:kappa}
\kappa:=\Theta^{-1}.
\end{equation}
Then both $\Theta$ and  $\kappa$ belong to class $ \mathcal K_\infty$. Additionally, $\kappa$ is equivalently given by the direct mapping from $\gamma$ as
\begin{equation}\label{eq:kappa_legendre}
\fbox{$\displaystyle \kappa(s)=(\ell\gamma)'\circ\left(\frac{\ell\gamma}{\mathrm{Id}}\right)^{-1}(s)$}\qquad s\geq 0\,,
\end{equation}
where $\ell$ denotes the Legendre-Fenchel transformation $\ell\gamma(s) = \int_0^s (\gamma'(\sigma))^{-1} d\sigma$ and ${\rm Id}(s)=s$.
\end{lemma}

\begin{proof}
Since $\gamma'\in \mathcal K_\infty$ and $\gamma(0)=0$, we have
\begin{equation}\label{eq:gamma_int}
\gamma(s)=\int_0^s \gamma'(\sigma)\,d\sigma,
\end{equation}
hence $\gamma\in \mathcal K_\infty$. For every $s>0$, strict increase of $\gamma'$ gives
\begin{equation}\label{eq:gamma_bound0}
\gamma(s)=\int_0^s \gamma'(\sigma)\,d\sigma < s\,\gamma'(s),
\end{equation}
so
\begin{equation}\label{eq:Theta_pos}
0<\Theta(s)<s \qquad \text{for all } s>0.
\end{equation}
Moreover,
\begin{equation}\label{eq:Theta_der}
\Theta'(s)=\frac{\gamma(s)\,\gamma''(s)}{(\gamma'(s))^2}>0 \qquad \text{for all } s>0,
\end{equation}
so $\Theta$ is strictly increasing on $(0,\infty)$, and hence on $[0,\infty)$ with $\Theta(0)=0$. If $\Theta$ were bounded above, there would exist $M>0$ such that
\begin{equation}\label{eq:Theta_bound}
s-\frac{\gamma(s)}{\gamma'(s)}\le M
\end{equation}
for all large $s$, that is,
\begin{equation}\label{eq:gamma_ratio}
\frac{\gamma'(s)}{\gamma(s)}\le \frac{1}{s-M}
\end{equation}
for all large $s$. Integrating \eqref{eq:gamma_ratio} yields at most linear growth of $\gamma$, contradicting the fact that $\gamma'\in \mathcal K_\infty$ is increasing and unbounded. Hence $\Theta(s)\to\infty$ as $s\to\infty$. Therefore $\Theta$ is continuous, strictly increasing, positive definite, and unbounded, so $\Theta\in \mathcal K_\infty$, and by invertibility \eqref{eq:kappa}, $\kappa\in \mathcal K_\infty$. Finally, we  verify that \eqref{eq:kappa_legendre} implies \eqref{eq:Theta}, \eqref{eq:kappa}. Let $z=\kappa(s)$, with $\kappa$ defined by \eqref{eq:kappa_legendre}. Then there exists $\rho$ such that $\rho=\bigl(\ell\gamma/\mathrm{Id}\bigr)^{-1}(s)$ and $z=(\ell\gamma)'(\rho)$. Since $(\ell\gamma)'=(\gamma')^{-1}$, we have $\rho=\gamma'(z)$. Therefore
\begin{equation}
s=\frac{\ell\gamma(\rho)}{\rho}=\frac{\rho z-\gamma(z)}{\rho}=z-\frac{\gamma(z)}{\gamma'(z)}=\Theta(z),
\end{equation}
that is, $\Theta(\kappa(s))=s$. Since $\Theta$ is strictly increasing, it follows that $\kappa=\Theta^{-1}$. 
\end{proof}

\begin{figure*}[t]
    \centering
\includegraphics[width=0.85\linewidth]{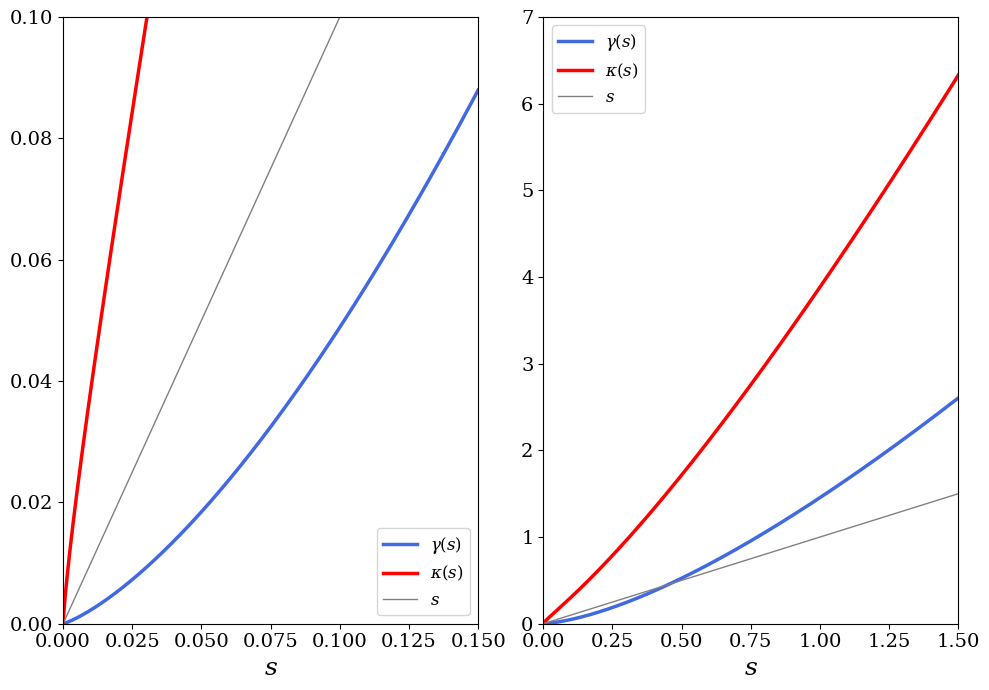}
\caption{Illustration of the expander from Example~\ref{ex1}. The left panel shows the flat
{\color{blue}$\gamma(s)\sim \frac{\alpha s}{\ln(1/s)}$}
resulting in the steep  
{\color{red}$\kappa(s)\sim s\ln(1/s)$} as $s\to 0^+$. 
The right panel shows 
{\color{blue}$\gamma(s)\sim s\ln s $} resulting in the superlinear  {\color{red}$\kappa(s)\sim s\ln s$} for large $s$. 
}   \label{fig:placeholder}
\end{figure*}

\begin{example}[\bf Locally and globally superlinear expander]\rm
\label{ex1}
To illustrate how the user can produce expanders that are not just linear and linear-in-portions functions, but truly nonlinear expanders, we give an example. The function $\gamma$ is chosen to induce superlinear behavior both for small and large $s$. Consider 
\begin{align}
\gamma(s)=s\left(\ln(1+s)+\frac{\alpha}{\ln(e+1/s)}\right),\qquad \alpha>0\,. 
\end{align}
Its own asymptotics are
\begin{equation}
\gamma(s)\sim \frac{\alpha s}{\ln(1/s)} \quad (s\to 0^+),
\qquad
\gamma(s)\sim s\ln s \quad (s\to\infty)\,, \label{eq:gamma-example-1}
\end{equation}
meaning that not only does $\gamma$ grow superlinearly for large $s$ but, near $0$, it dominates (it is ``less flat than'') every power $s^p$ with $p>1$.
After computing
\begin{equation}
    \Theta(s) = s - \frac{s\left(\ln(1+s) + \dfrac{\alpha}{A(s)}\right)}{\ln(1+s) + \dfrac{s}{1+s} + \alpha\left(\dfrac{1}{A(s)} +  \dfrac{1}{(es+1)A(s)^2} \right)}\,, 
\end{equation}
\begin{align}
\Theta(s)=\kappa^{-1}(s)
&= \frac{s^2 B(s)}
{\ln(1+s)+\dfrac{\alpha}{A(s)}+sB(s)}\,, \\ 
A(s) &:= \ln\left(e+\frac{1}{s}\right), \\
B(s) &:= \frac{1}{1+s}+\frac{\alpha}{(es+1)A(s)^2}, \
\end{align}
the  expander $\kappa$ satisfies
\begin{equation} \label{eq:kappa-example-1}
\kappa(s)\sim s\ln(1/s)\quad (s\to 0^+),
\qquad
\kappa(s)\sim s\ln s\quad (s\to\infty)\,,
\end{equation}
meaning  not only that $\kappa$ has a superlinear growth at infinity, but also at zero,
\begin{equation}
\kappa'(0^+)=+\infty\,.
\end{equation}
\mbox{}\hfill$\triangle$
\end{example}

\section{Control-Cost-Parametrized Half-Direct-Optimal Stabilizing Feedback Laws}
\label{sec:half-direct-optimal-design}
Given a CLF and a cost function $\gamma$ on the control input, the following theorem constructs an explicit stabilizing feedback law that is simultaneously the minimizer of an infinite-horizon cost functional parametrized by $\gamma$.

 \begin{theorem}[General half-direct-optimal stabilizer family]\label{thm:main}
Consider the control--affine system
\begin{equation}\label{eq:system}
\dot x = f(x) + g(x)u\,, \qquad x \in\mathbb{R}^n
\end{equation}
with a $C^1$ and proper function $V$ such that $V(0)=0$, $V(x)>0$ for all $x\neq 0$ 
and satisfying the Control Lyapunov Function (CLF) property
\begin{equation}\label{eq:clf}
L_g V(x) = 0,\ x \neq 0 \;\Longrightarrow\; L_f V(x)  
\leq -W(x)\,,
\end{equation}
where $W(x)$ is a continuous, positive definite, and proper function. Define the nominal feedback using the Freeman-Kokotovic minimum-norm formula \cite{doi:10.1137/S0363012993258732}
\begin{equation}\label{eq:u0}
u_0(x):=
\begin{cases}
-\dfrac{\max\{0,L_fV(x)+W(x)\}}{L_gV(x)}, & L_gV(x)\neq 0,\\[1ex]
0, & L_gV(x)=0\,.
\end{cases}
\end{equation}
Then, for the feedback
\begin{equation}\label{eq:ustar}
\fbox{$u^*(x) := \operatorname{sgn}(u_0(x)) \, \kappa\!\big(|u_0(x)|\big)$}
\end{equation}
where $\kappa$ is defined in \eqref{eq:kappa} for a function $\gamma$ satisfying the conditions in Lemma \ref{lem:Theta_kappa_Kinf}, the following results hold:
\begin{enumerate}
\item[(1)] the feedback $u^*(x)$ is the unique minimizer of the infinite-horizon cost functional
\begin{equation}\label{eq:J}
J(u) = \int_0^\infty \left[ q(x(t)) + 
r(x(t))\, 
\gamma\!\left(
|u(t)|\right) \right]\, dt,
\end{equation}

\item[(2)] the state cost and control weight are
\begin{equation}\label{eq:q}
q(x) = \max\{W(x),-L_fV(x)\} > 0
\end{equation}
\begin{equation}\label{eq:r}
r(x) =
\frac{|L_gV(x)|}
{\left(\dfrac{\ell\gamma}{\mathrm{Id}}\right)^{-1}
\!\left(
\dfrac{\max\{0,L_fV(x)+W(x)\}}{|L_gV(x)|}
\right)}
> 0
 \end{equation}
for all $x\neq 0$ and $q(0)=0$, $r(0)=+\infty$. 

\item[(3)] the minimal value is $V(x_0)$,   

\item[(4)] both $u^*(x)$ and the associated nominal control $u_0(x)$ are globally asymptotically stabilizing, and 

\item[(5)] both $u^*$ and $u_0$ are continuous in $x$ provided $(f,g,V,W)$ are such that $
\frac{L_fV(x)+W(x)}{L_gV(x)}\to 0$ as $ x\to x_0$ whenever $L_gV(x_0)=0$.
\end{enumerate}

\end{theorem}

\begin{proof}
First, define
\begin{equation}\label{eq:sigma}
\sigma(x) := \frac{|L_g V(x)|}{
{r(x)}}.
\end{equation}
By construction of $r(x)$ in \eqref{eq:r}, one has
\begin{equation}\label{eq:legendre_match}
\frac{\ell\gamma\big(\sigma(x)\big)}{\sigma(x)} = \frac{\max\{0,L_f V(x) + W(x)\}}{|L_gV(x)|},
\end{equation}
hence
\begin{equation}\label{eq:W_positive}
r(x)\ell\gamma\big(\sigma(x)\big) - L_f V(x) = q(x) > 0 \quad \text{for } x \neq 0.
\end{equation}
Define the Hamiltonian
\begin{equation}\label{eq:H}
H(x,u) = L_f V(x) + L_g V(x)\,u + q(x) + 
r(x) \gamma\!\big(|u|\big).
\end{equation}
The choice of $u^*(x)$ corresponds to the pointwise minimizer of $H$ with respect to $u$. Indeed, for fixed $x$, the map
\begin{equation}\label{eq:convex_map}
u \mapsto L_g V(x)\,u + 
r(x)\,
\gamma\!\big(|u|\big)
\end{equation}
is strictly convex, and its stationarity condition is
\begin{equation}\label{eq:stationarity}
0=L_g V(x)+\operatorname{sgn}(u)\,{r(x)}\,\gamma'\!\big(|u|\big).
\end{equation}
Thus, with
\begin{equation}\label{eq:ustar_magnitude}
|u^*(x)|=(\gamma')^{-1}\!\big(\sigma(x)\big)=(\ell\gamma)'\big(\sigma(x)\big),
\end{equation}
one gets
\begin{equation}\label{eq:ustar_explicit}
u^*(x)
=
-\operatorname{sgn}\!\big(L_gV(x)\big)\,
{(\ell\gamma)'\big(\sigma(x)\big)}
=
\operatorname{sgn}(u_0(x))\,\kappa\!\big(|u_0(x)|\big),
\end{equation}
where the last equality follows from
\begin{equation}\label{eq:u0_relation}
|u_0(x)|
=\frac{\max\{0,L_fV(x)+W(x)\}}{|L_gV(x)|}
=\frac{\ell\gamma(\sigma(x))}{\sigma(x)}
\end{equation}
and \eqref{eq:kappa_legendre}. 
Optimality of $u^*$ follows from the fact that \eqref{eq:convex_map} is strictly convex in $u$ and that $u^*$ satisfies the stationarity condition \eqref{eq:stationarity}, hence is the pointwise-in-$x$ minimizer of the Hamiltonian, 
\begin{equation}\label{eq:H_opt}
H(x,u)\ge H(x,u^*(x)) = 0\,.
\end{equation}
Substitution yields
\begin{equation}\label{eq:Vdot}
\dot V(x) = L_f V(x) + L_g V(x)\,u^*(x)
= -q(x) - 
r(x)\,
\gamma\!\big(|u^*(x)|\big) < 0
\quad \text{for } x \neq 0.
\end{equation}
Thus $V$ is a strict Lyapunov function under $u^*$, implying global asymptotic stability. For $u_0$, global asymptotic stability holds as well because 
\begin{equation}
\dot V(x)=
\begin{cases}
-\max\{W(x),-L_fV(x)\}, & L_gV(x)\neq 0,\\[4pt]
L_fV(x), & L_gV(x)=0.
\end{cases}
\end{equation}
\end{proof}

\bigskip

\begin{example}[\bf Enhancement of convergence rate with expander]\rm
\label{ex2}
We return to Example \ref{ex1}. We illustrate the effect of the nonlinear expander $\kappa$ (from Example \ref{ex1}) by considering stabilization of the system $\dot{x}=u$ using the nominal controller $u_0=-x$ and the modified controller $u^*=-\kappa(x)$, and comparing the resulting closed-loop dynamics $\dot{x}=-x$ and $\dot{x}=-\kappa(x)$ from a common initial condition $x_0$. The trajectories under $-\kappa(x)$ decay significantly faster across both large and small state regimes, which is quantified by the convergence time to a prescribed level $\varepsilon\in(0,1)$: 
\begin{equation}
T_{\rm lin}(x_0,\varepsilon)=\ln(x_0/\varepsilon)\,, \qquad 
T_\kappa(x_0,\varepsilon)
=
\ln\!\bigl(\ln x_0\,\ln(1/\varepsilon)\bigr)+O(1),
\end{equation}
where the latter is obtained from $T_\kappa(x_0,\varepsilon)=\int_\varepsilon^{x_0} d\xi/\kappa(\xi)$, 
showing that the dependence on both the initial condition and the target level is reduced from logarithmic to double-logarithmic order. In summary, $T_\kappa(x_0,\varepsilon)\le 2\ln T_{\rm lin}(x_0,\varepsilon)+O(1)$ as $T_{\rm lin}(x_0,\varepsilon)\to\infty$.
\mbox{}\hfill$\triangle$
\end{example}

\bigskip

Next, we apply the expander $\kappa$ to the half-Sontag formula, generalizing it beyond the classical full Sontag formula for $\kappa(s) = 2s$ and $\gamma(s) = s^2$. The cost on the state $q(x)$ remains unchanged relative to the inverse optimality of Sontag's formula, but the cost and weight on control change significantly. 

\begin{corollary}[From half-Sontag formula to an expander-parametrized family of universal formulas]\label{cor:half_sontag}
Let
\begin{equation}\label{eq:uS2}
u_{S/2}(x):=
\begin{cases}
-\dfrac{L_fV(x)+\sqrt{(L_fV(x))^2+|L_gV(x)|^4}}{2\,L_gV(x)}, & L_gV(x)\neq 0,\\[1.2ex]
0, & L_gV(x)=0\,,
\end{cases}
\end{equation}
where the $C^1$, positive definite, proper $V(x)$ satisfies the CLF condition $L_gV(x) \Rightarrow L_fV(x)<0$ whenever $x\neq 0$. Define
\begin{equation}\label{eq:uKS}
u^*_{\kappa(S/2)}(x) =
\operatorname{sgn}\!\big(u_{S/2}(x)\big)\,
\kappa\!\big(|u_{S/2}(x)|\big).
\end{equation}
Then all statements of Theorem~\ref{thm:main} hold with
\begin{equation}\label{eq:qS2}
q(x) =
\frac{|L_g V(x)|^4}{2\left( L_f V(x)+ \sqrt{(L_f V(x))^2 + |L_g V(x)|^4}
\right)},
\end{equation}
\begin{equation}
\label{eq:rS}
r(x)=
\frac{|L_gV(x)|}
{\left(\dfrac{\ell\gamma}{\mathrm{Id}}\right)^{-1}
\!\left(
\dfrac{L_fV(x)+\sqrt{(L_fV(x))^2+|L_gV(x)|^4}}{2|L_gV(x)|}
\right)}.
\end{equation}
If, in addition, $(f,g,V)$ satisfy the small-control property at $x=0$, namely, for every $\varepsilon>0$ there exists $\delta>0$ such that, whenever $|x|<\delta$, there exists $u$ with $|u|<\varepsilon$ and $L_fV(x)+L_gV(x)\,u<0$, then $u_{S/2}$ and $u^*_{\kappa(S/2)}$ are continuous in $x$.
\end{corollary}

\begin{proof}
The half-Sontag choice corresponds to selecting
\begin{equation}\label{eq:W_half}
W(x) =
\frac{\sqrt{(L_f V(x))^2 + |L_g V(x)|^4} - L_f V(x)}{2},
\end{equation}
which, substituted into the Freeman-Kokotovic formula \eqref{eq:u0} results in 
\begin{equation}
\max\{0,L_fV+W\}
=
\frac{L_fV+\sqrt{(L_fV)^2+|L_gV|^4}}{2}.
\end{equation}
Substituting \eqref{eq:W_half} into Theorem~\ref{thm:main} yields exactly \eqref{eq:uS2}, \eqref{eq:qS2}, and \eqref{eq:rS}.
\end{proof}

\bigskip
Next, we make two remarks.

\paragraph{State-Cost-Parametrized Half-Direct Optimal Control.} In addition to the control-cost-parametrized half-direct optimal control problem \eqref{eq:J} addressed in this paper, one may consider its state-cost-parametrized dual, in which the running cost
\begin{equation}
\label{eq-half-dir-state}
J(u) = \int_0^\infty \left[ q(x(t)) + 
r(x(t))\,u(t)^2 \right]\,dt
\end{equation}
is prescribed through a nonnegative function $q$, and a feedback law is sought that is inverse optimal for this cost for some $r>0$. Restricting attention to scalar input and feedback laws of the form
\begin{equation}
\label{eq-LgV-dual}
u(x) = - \pi (x)\, b(x), \qquad b(x) = L_g V(x),
\end{equation}
the problem reduces pointwise to the algebraic mapping
\begin{equation}
(a, b, q) \;\longmapsto\; r\,,
\end{equation}
where $a = L_f V$, subject to the feasibility condition
\begin{equation}
q(x) \geq \max\{0, - 4 a(x)\}\,,
\end{equation}
and with the induced quantities given explicitly by
\begin{eqnarray}
\label{eq-FK-gain}
\pi &=& 2\frac{a+q}{ b^2}>0
\\
\label{eq-FK-control-weight}
r &=& \frac{1}{2\pi} = \frac{ b^2}{4(a+ q)}>0.
\end{eqnarray}
All relations are understood pointwise in the state. 
The feedback \eqref{eq-LgV-dual}, \eqref{eq-FK-gain} is of the Freeman-Kokotovic (FK) form but, unlike the original min-norm control which is only pointwise optimal, \eqref{eq-LgV-dual}, \eqref{eq-FK-gain} doubles the FK formula and is infinite-horizon half-direct optimal. 
Thus, in the formulation that is dual to the half-direct optimal problem solved in the rest of the paper, the design problem we mention reduces to the selection of a single function $q(x)$ subject to a simple pointwise-in-$x$ inequality constraint, with the resulting feedback gains obtained by direct algebraic evaluation at each $x$. While this formulation offers a broader space of admissible control law choices, the design entails no operator structure (in addition to offering limited guidance on how the chosen function $q$ affects closed-loop performance). For this reason, and in contrast to the ``control-side'' formulation developed in this paper, which gives rise to design via operator mapping from a scalar function $\gamma$ to a feedback law, the ``state-side'' half-direct problem is less interesting and not pursued here.

\paragraph{HJB PDEs explicitly solved by CLF $V(x)$.} Both of the half-direct optimal control formulations avoid the need to solve the HJB PDE, since the starting CLF $V(x)$ is already the solution, in the case of \eqref{eq:J}, of the PDE
\begin{equation}
L_fV(x)-r(x)\ell\gamma\!\left(\frac{|L_gV(x)|}{{r(x)}}\right) +q(x)=0\,
\end{equation}
for a given $\gamma(s)$, and in the case of \eqref{eq-half-dir-state} of the PDE
\begin{equation}
L_fV(x)  - \frac{(L_gV(x))^2}{4\,r(x)}+ q(x)=0\,
\end{equation}
for a given $q(x)$. 

\section{Lipschitzness and Neural Approximation of the Mapping $\gamma \mapsto \kappa$} \label{sec:lipschitz-approx}

For the family of controllers in Theorem \ref{thm:main} to be useful at scale --- whether for offline exploration of the design space or for real-time adaptation --- the map $\gamma \mapsto \kappa$ must be well-behaved as $\gamma$ varies. We establish that it is Lipschitz, and that this regularity is sufficient to guarantee uniform neural operator approximation over compact families of cost functions.

\begin{lemma}\label{lem:kappa_lipschitz}
Let $K=[\underline z,\overline z]\subset(0,\infty)$ and let $\mathcal G\subset C^2(K)$ be compact in $C^1(K)$. Assume that there exist constants $\mu>0$, $M>0$, and $m>0$ such that, for every $\gamma\in\mathcal G$ and every $z\in K$,
\begin{eqnarray}
\gamma(z) &\le& M, \label{eq:gamma_bound}\\
\gamma'(z) &\ge& \mu, \label{eq:Dgamma_bound}\\
\frac{\gamma(z)\,\gamma''(z)}{(\gamma'(z))^2} &\ge& m. \label{eq:nondegeneracy}
\end{eqnarray}
Then the map 
\begin{equation}
\mathcal{K}:\gamma\mapsto \kappa
\end{equation}
where $\kappa$ is defined by \eqref{eq:kappa_legendre}, or, alternatively, by \eqref{eq:kappa}
, \eqref{eq:Theta}, is Lipschitz from $C^1(K)$ into $L^\infty(\Theta(K))$. More precisely, for any $\gamma_1,\gamma_2\in\mathcal G$,
\begin{equation}\label{eq:kappa_lip}
\boxed{\displaystyle\|\kappa_1-\kappa_2\|_{L^\infty(\Theta(K))}
\le
\frac{1}{\mu m}
\left(
\|\gamma_1-\gamma_2\|_{L^\infty(K)}
+
\frac{M}{\mu}\|\gamma'_1-\gamma'_2\|_{L^\infty(K)}
\right)}
\end{equation}
\end{lemma}

\begin{proof}
For $i=1,2$, let
\begin{equation}\label{eq:Theta_i}
\Theta_i(z)=z-\frac{\gamma_i(z)}{\gamma'_i(z)},\qquad
\kappa_i=\Theta_i^{-1}.
\end{equation}
By \eqref{eq:nondegeneracy},
\begin{equation}\label{eq:Dtheta_bound}
\Theta'_i(z)=\frac{\gamma_i(z)\,\gamma''_i(z)}{(\gamma'_i(z))^2}\ge m>0
\quad \text{for all } z\in K,
\end{equation}
so each $\Theta_i$ is strictly increasing on $K$. Fix $s\in \Theta(K)$ and set $z_i=\kappa_i(s)$. Then
\begin{equation}\label{eq:theta_equal}
\Theta_1(z_1)=\Theta_2(z_2),
\end{equation}
and hence
\begin{equation}\label{eq:theta_diff}
\Theta_1(z_1)-\Theta_1(z_2)=\Theta_2(z_2)-\Theta_1(z_2).
\end{equation}
By the mean value theorem, there exists $\xi$ between $z_1$ and $z_2$ such that
\begin{equation}\label{eq:mvt}
\Theta'_1(\xi)(z_1-z_2)=\Theta_2(z_2)-\Theta_1(z_2).
\end{equation}
Using \eqref{eq:Dtheta_bound},
\begin{equation}\label{eq:z_bound}
|z_1-z_2|
\le
\frac{1}{m}\|\Theta_1-\Theta_2\|_{L^\infty(K)}.
\end{equation}
Taking the supremum over $s\in \Theta(K)$ yields
\begin{equation}\label{eq:kappa_theta}
\|\kappa_1-\kappa_2\|_{L^\infty(\Theta(K))}
\le
\frac{1}{m}\|\Theta_1-\Theta_2\|_{L^\infty(K)}.
\end{equation}
From \eqref{eq:Theta_i},
\begin{equation}\label{eq:theta_diff2}
\Theta_1(z)-\Theta_2(z)
=
-\frac{\gamma_1(z)}{\gamma'_1(z)}+\frac{\gamma_2(z)}{\gamma'_2(z)}.
\end{equation}
Therefore,
\begin{equation}\label{eq:theta_est1}
|\Theta_1(z)-\Theta_2(z)|
\le
\frac{|\gamma_1(z)-\gamma_2(z)|}{|\gamma'_1(z)|}
+
|\gamma_2(z)|
\left|
\frac{1}{\gamma'_1(z)}-\frac{1}{\gamma'_2(z)}
\right|.
\end{equation}
Using \eqref{eq:Dgamma_bound} and \eqref{eq:gamma_bound},
\begin{equation}\label{eq:theta_est2}
|\Theta_1(z)-\Theta_2(z)|
\le
\frac{1}{\mu}|\gamma_1(z)-\gamma_2(z)|
+
\frac{M}{\mu^2}|\gamma'_1(z)-\gamma'_2(z)|.
\end{equation}
Taking the supremum over $z\in K$ and combining with \eqref{eq:kappa_theta} yields \eqref{eq:kappa_lip}.
\end{proof}

\begin{proposition}\label{thm:operator_uat}
Let $K=[\underline z,\overline z]\subset(0,\infty)$ and let $\mathcal G\subset C^2(K)$ be compact in $C^1(K)$. Under the assumptions of Lemma~\ref{lem:kappa_lipschitz}, for every $\varepsilon>0$ there exists a neural operator $\widehat{\mathcal K}$ such that
\begin{equation}\label{eq:uat}
\sup_{\gamma\in\mathcal G}
\|\widehat{\mathcal K}(\gamma)-\mathcal K(\gamma)\|_{L^\infty(\Theta(K))}
<\varepsilon.
\end{equation}
\end{proposition}

\begin{proof}
The conclusion follows from Lemma~\ref{lem:kappa_lipschitz} and a universal approximation theorem for neural operators on compact subsets of $C^1(K)$ \cite{lu2021Learning}.
\end{proof}

\section{Stability under Approximated Half-Direct-Optimal Feedback}\label{sec:spas}

Naturally, the half-direct-optimal feedback law then becomes 
\begin{align}
    \hat{u}^\ast(x) :=&\; \mathrm{sgn}(u_0(x)) \hat{\kappa}(|u_0(x)|)\,, \label{eq:approximate-feedback-law}\\
    \hat{\kappa}(\cdot) =&\; \widehat{\mathcal{K}}(\gamma)(\cdot)
\end{align}
where $\hat{\kappa}$ is the learned neural approximator. 

We now aim to answer two questions: What type of stability is achieved under the approximated feedback map and, equally important, is the corresponding law practically inverse optimal? 
\begin{theorem}[Stability and half-direct optimality under the neural operator learned controller] 
\label{thm:semi-global}
Let the assumptions of Lemma~\ref{lem:kappa_lipschitz} and Theorem~\ref{thm:main} hold. Consider the neural operator approximation satisfying Proposition~\ref{thm:operator_uat} with the approximated feedback law $\hat{u}^\ast$ as in \eqref{eq:approximate-feedback-law}. 
Then the following hold.

\begin{enumerate}
        \item \textbf{Semiglobal practical asymptotic stability.}
    For every pair $R_0>\delta>0$, define the level set $\mathcal S_R:=\{x:V(x)\le R\}$ where $R := \alpha_2(R_0)$ and $\alpha_2 \in \mathcal{K}_\infty$ is such that $V(x) \leq \alpha_2(|x|)$. 
    Then there exist $\varepsilon^\ast(R_0,\delta)>0$ and $\beta_{R_0,\delta}\in\mathcal{KL}$ such that, if the initial condition is constrained to the ball of radius $R_0$:
    \begin{align}
     x(0)\in \mathbb{B}_{R_0}\,, \qquad \mathbb{B}_{R_0} = \{x \in \mathbb{R}^n:|x| \le R_0\}\,, 
    \end{align}
    and the approximation $\widehat{\mathcal{K}}$ guaranteed to exist by Proposition \ref{thm:operator_uat} satisfies the uniform error  
    \begin{align}
         \varepsilon\in(0,\varepsilon^\ast(R_0,\delta)),
    \end{align}
    then, for all $t \ge 0$, the corresponding closed-loop trajectory remains in $\mathcal S_R$ and satisfies
    \begin{align}
    \boxed{| x(t)| \le \max\bigl\{\beta_{R_0,\delta}(| x(0)|,t),\,\delta\bigr\}}
    \end{align}

        \item \textbf{Induced state-cost perturbation.}
    Let $x \in \mathcal S_R$ be any fixed state, $q=W$ denote the nominal state cost from Theorem~\ref{thm:main}, and define the induced state cost associated with $\hat{u}^\ast$ by
    \begin{align}\label{eq:qhat_def}
        \hat q(x):=
        -L_fV(x)-L_gV(x)\,\hat u^\ast(x)
        -
        \gamma\!\bigl(\sqrt{r(x)}\,|\hat u^\ast(x)|\bigr).
    \end{align}
    Then, $\hat q$ is a second-order perturbation of the nominal state cost $q$. Namely, there exists $C_R>0$ such that
    \begin{align}\label{eq:cost_perturbation_bound}
        \sup_{x\in\mathcal S_R}\bigl|q(x)-\hat q(x)\bigr|
        \le C_R\,\varepsilon^2.
    \end{align}
\end{enumerate}
\end{theorem}

Theorem~\ref{thm:semi-global} shows that the learned feedback $\hat u^\ast$ preserves the stabilizing behavior of the nominal half-direct-optimal law on a semiglobal $\epsilon$ practical set. Second, it provides insight into the \emph{pointwise half-direct-optimality distortion estimate}. More precisely, although the control-law error is of order $\mathcal{O}(\varepsilon)$, for any fixed state, the induced state-cost distortion satisfies $q(x)-\hat q(x)=\mathcal{O}(\varepsilon^2)$. Note, this distortion is a direct reflection of the Hamiltonian residual error incurred by using the learned control $\hat{u}^\ast$ instead of the minimizer $u^\ast$. The proportionality constant $C_R$ depends on both $r(x)$ and the curvature $\gamma''$, and can become large in regions where the control weight $r(x)$ is high and $\gamma$ is sharply curved, so choosing a $\gamma$ with a moderately growing second derivative keeps $C_R$ manageable and ensures the quadratic suppression of the approximation error is genuinely beneficial. Lastly, we note this does not specify the specific difference of the entire trajectory from $x_0$ following $\hat{u}^\ast$. To do so, one must analyze the finite-horizon optimality gap, since, due to SPAS in Theorem \ref{thm:semi-global}, the infinite-horizon cost-function $J$ can be infinite. Hence, for the practitioner, we provide the following finite horizon result:

\begin{proposition}[Finite-horizon near-optimality with terminal penalty] \label{corr:finite-horizon}
Fix $R\in(0,R^\ast]$ and $T>0$. Let $x^\ast(\cdot)$ and $\hat x(\cdot)$ denote the closed-loop trajectories generated by $u^\ast$ and $\hat u^\ast$, respectively, with the same initial condition $x_0\in\mathcal S_R$. Define
\begin{align}
    J_T(u^\ast;x_0)
    &:=
    \int_0^T
    \Big(
        q(x^\ast(t))
        +
        \gamma\!\bigl(\sqrt{r(x^\ast(t))}\,|u^\ast(x^\ast(t))|\bigr)
    \Big)\,dt\,, \\
     J_T(\hat u^\ast;x_0)
    &:=
    \int_0^T
    \Big(
        q(\hat x(t))
        +
        \gamma\!\bigl(\sqrt{r(\hat x(t))}\,|\hat u^\ast(\hat x(t))|\bigr)
    \Big)\,dt.
\end{align}
Then
\begin{align}
    \bigl[J_T(\hat u^\ast;x_0)+V(\hat x(T))\bigr]
    -
    \bigl[J_T(u^\ast;x_0)+V(x^\ast(T))\bigr]
    \le
    C_R\,\varepsilon^2\,T.
\end{align}
\end{proposition}

\section{Proofs of Theorem \ref{thm:semi-global} and  Proposition \ref{corr:finite-horizon}} \label{sec:proofs}
\subsection{Proof of Theorem \ref{thm:semi-global}}
\begin{proof}
\textbf{Part (1)}
Since $V$ and $W$ are continuous, positive definite, and proper, there exist $\alpha_3,\alpha_4\in\mathcal \mathcal K_\infty$ such that
\begin{align}
\alpha_3(|x|)\le V(x), \qquad \alpha_4(|x|)\le W(x), \qquad \forall x\in\mathbb R^n.
\end{align}
Further, using Proposition \ref{thm:operator_uat}, on the set $\mathcal{S}_{R}$, we have
\begin{align}
|\hat u^\ast(x)-u^\ast(x)|\le \varepsilon, \qquad \forall x\in\mathcal S_R.
\end{align}
Since $L_gV$ is continuous and $\mathcal S_R$ is compact, define
\begin{align}
M_{R_0}:=\sup_{x\in\mathcal S_R}|L_gV(x)|<\infty.
\end{align}
Along the closed-loop trajectory $ x(\cdot)$, adding and subtracting $L_gVu^\ast(x)$, by definition of $W$, and using $V(x) \leq \alpha_2(|X|)$ yields
\begin{align}
\nonumber \dot V(x)
&=L_fV(x)+L_gV(x)\hat u^\ast(x) \\
\nonumber&=\bigl(L_fV(x)+L_gV(x)u^\ast(x)\bigr)
+L_gV(x)\bigl(\hat u^\ast(x)-u^\ast(x)\bigr) \\
\nonumber&=-W(x)-
\gamma\!\bigl(\sqrt{r(x)}\,|u^\ast(x)|\bigr)
+L_gV(x)\bigl(\hat u^\ast(x)-u^\ast(x)\bigr) \\
\nonumber&\le -W(x)+|L_gV(x)|\,|\hat u^\ast(x)-u^\ast(x)| \\
\nonumber&\le -W(x)+M_{R_0}\varepsilon \\
\nonumber&\le -\alpha_4(|x|)+M_{R_0}\varepsilon \\ 
&\le -\alpha_4(\alpha_2^{-1}(V(\hat{x}))) + M_{R_0}\varepsilon\,. 
\end{align}
By choosing $\varepsilon^\ast$ such that
\begin{align}
\varepsilon^\ast(R_0,\delta):=
\frac{1}{2M_{R_0}}
\min\bigl\{\alpha_4(\alpha_2^{-1}(\alpha_3(\delta))),\,\alpha_4(R_0)\bigr\}.
\end{align}
then, whenever $\varepsilon \in (0, \varepsilon^\ast)$ and  $V(x)\ge \alpha_3(\delta)$, we have 
\begin{align}
\dot V(x)
& \nonumber \le -\alpha_4(\alpha_2^{-1}(V(x)))+M_{R_0}\varepsilon \\
& \nonumber \le -\alpha_4(\alpha_2^{-1}(V(x)))
+\frac{1}{2}\alpha_4(\alpha_2^{-1}(\alpha_3(\delta))) \\
&\le -\frac{1}{2}\alpha_4(\alpha_2^{-1}(V(x))), \label{eq:lyapunov-bound}
\end{align}
where the last inequality follows from the monotonicity of $\alpha_4\circ\alpha_2^{-1}$.
Moreover, on the boundary of $\mathcal S_R$, that is, when $V(x)=\alpha_2(R_0)$, we have
\begin{align}
\dot V(x)
&\le -\alpha_4(R_0)+M_{R_0}\varepsilon \nonumber \\
&<0.
\end{align}
Hence $\mathcal S_R$ is forward invariant.
By the comparison lemma and \eqref{eq:lyapunov-bound}, there exists $\widetilde\beta\in\mathcal{KL}$ such that
\begin{align}
V(x(t))
\le
\max\bigl\{\widetilde\beta(V(x(0)),t),\,\alpha_3(\delta)\bigr\},
\qquad \forall t\ge 0.
\end{align}
Using
\begin{align}
V(x(0))\le \alpha_2(|x(0)|),
\qquad
\alpha_3(|x(t)|)\le V(x(t)),
\end{align}
it follows that
\begin{align}
|x(t)|
&\le
\alpha_3^{-1}
\Bigl(
\max\bigl\{\widetilde\beta(\alpha_2(|x(0)|),t),\,\alpha_3(\delta)\bigr\}
\Bigr) \\
&=
\max\Bigl\{
\alpha_3^{-1}\bigl(\widetilde\beta(\alpha_2(|x(0)|),t)\bigr),
\,\delta
\Bigr\}.
\end{align}
Therefore, defining
\begin{align}
\beta_{R_0,\delta}(r,t):=
\alpha_3^{-1}\bigl(\widetilde\beta(\alpha_2(r),t)\bigr),
\end{align}
we obtain
\begin{align}
|x(t)|\le \max\bigl\{\beta_{R_0,\delta}(|x(0)|,t),\,\delta\bigr\},
\qquad \forall t\ge 0.
\end{align}
This completes the result.

\textbf{Part (2).}
Fix $x\in\mathcal S_R$. Since $q=W$ by Theorem~\ref{thm:main}, and since the Hamiltonian identity $H(x,u^\ast(x))=0$ yields
\begin{align}
    q(x)
    =
    -L_fV(x)-L_gV(x)\,u^\ast(x)
    -
    \gamma\!\bigl(\sqrt{r(x)}\,|u^\ast(x)|\bigr),
\end{align}
subtracting~\eqref{eq:qhat_def} gives
\begin{align}
    q(x)-\hat q(x)
    =
    \varphi_x(\hat u^\ast(x))-\varphi_x(u^\ast(x)),
\end{align}
where
\begin{align}
    \varphi_x(u)
    :=
    L_gV(x)\,u+
    \gamma\!\bigl(\sqrt{r(x)}\,|u|\bigr).
\end{align}
Since $u^\ast(x)$ minimizes $\varphi_x$, the stationarity condition~\eqref{eq:stationarity} gives
\begin{align}
    \varphi_x'(u^\ast(x))=0.
\end{align}
Moreover, $\varphi_x$ is $C^2$ and strictly convex. Hence, by Taylor's theorem, there exists $\xi(x)$ between $u^\ast(x)$ and $\hat u^\ast(x)$ such that
\begin{align}
    q(x)-\hat q(x)
    =
    \frac12\,r(x)\,\gamma''\!\bigl(\sqrt{r(x)}\,|\xi(x)|\bigr)\,
    \bigl(\hat u^\ast(x)-u^\ast(x)\bigr)^2\,. 
\end{align}
Since $|\hat u^\ast(x)-u^\ast(x)|\le\varepsilon$ on $\mathcal S_R$ and $\mathcal S_R$ is compact, we obtain 
\begin{align}
    \sup_{x\in\mathcal S_R}|q(x)-\hat q(x)|
    \le&\; C_R\,\varepsilon^2\,, \\
    C_R :=&\; \sup_{x \in \mathcal{S}_R} \frac{1}{2}r(x)\gamma''\left(\sqrt{r(x)}|\xi(x)|\right)\,. 
\end{align}
\end{proof}

\subsection{Proof of Proposition \ref{corr:finite-horizon}}
\begin{proof}
For the exact half-direct-optimal controller $u^\ast$, Theorem~\ref{thm:main} gives the dissipation identity
\begin{align}
    \dot V(x^\ast(t))
    =
    -q(x^\ast(t))
    -
    \gamma\!\bigl(\sqrt{r(x^\ast(t))}\,|u^\ast(x^\ast(t))|\bigr).
\end{align}
Integrating from $0$ to $T$ yields
\begin{align}
    V(x^\ast(T)) - V(x_0)
    =
    -\int_0^T
    \Big(
        q(x^\ast(t))
        +
        \gamma\!\bigl(\sqrt{r(x^\ast(t))}\,|u^\ast(x^\ast(t))|\bigr)
    \Big)\,dt,
\end{align}
that is,
\begin{align}
    J_T(u^\ast;x_0)+V(x^\ast(T)) = V(x_0).
\end{align}
For the learned controller $\hat u^\ast$, the induced dissipation identity gives
\begin{align}
    \dot V(\hat x(t))
    =
    -\hat q(\hat x(t))
    -
    \gamma\!\bigl(\sqrt{r(\hat x(t))}\,|\hat u^\ast(\hat x(t))|\bigr).
\end{align}
Adding and subtracting $q(\hat x(t))$ on the right-hand side,
\begin{align}
    \dot V(\hat x(t))
    &=
    -q(\hat x(t))
    -
    \gamma\!\bigl(\sqrt{r(\hat x(t))}\,|\hat u^\ast(\hat x(t))|\bigr)
    +\bigl(q(\hat x(t))-\hat q(\hat x(t))\bigr).
\end{align}
Integrating from $0$ to $T$ gives
\begin{align}
    V(\hat x(T)) - V(x_0)
    =&\;
    -\int_0^T
    \Big(
        q(\hat x(t))
        +
        \gamma\!\bigl(\sqrt{r(\hat x(t))}\,|\hat u^\ast(\hat x(t))|\bigr)
    \Big)\,dt
    \nonumber \\ &\;+\int_0^T \bigl(q(\hat x(t))-\hat q(\hat x(t))\bigr)\,dt.
\end{align}
Hence,
\begin{align}
    J_T(\hat u^\ast;x_0)+V(\hat x(T))
    =
    V(x_0)+\int_0^T \bigl(q(\hat x(t))-\hat q(\hat x(t))\bigr)\,dt.
\end{align}

Since $\hat x(0)\in\mathcal S_R$ and $\mathcal S_R$ is forward invariant under $\hat u^\ast$, we have $\hat x(t)\in\mathcal S_R$ for all $t\in[0,T]$. Therefore, by \eqref{eq:cost_perturbation_bound},
\begin{align}
    \int_0^T \bigl(q(\hat x(t))-\hat q(\hat x(t))\bigr)\,dt
    \le
    \int_0^T C_R\,\varepsilon^2\,dt
    =
    C_R\,\varepsilon^2\,T.
\end{align}
Thus,
\begin{align}
    J_T(\hat u^\ast;x_0)+V(\hat x(T))
    \le
    V(x_0)+C_R\,\varepsilon^2\,T.
\end{align}
Subtracting the identity for $u^\ast$ gives
\begin{align}
    \bigl[J_T(\hat u^\ast;x_0)+V(\hat x(T))\bigr]
    -
    \bigl[J_T(u^\ast;x_0)+V(x^\ast(T))\bigr]
    \le
    C_R\,\varepsilon^2\,T.
\end{align}
\end{proof}

\section{Numerical Examples}
\label{sec:num}

\subsection{Control Design}

Consider the unicycle dynamics governed by
\begin{subequations}
\label{eq:unicycle-system}
\begin{align}
    \dot{x} =\;& v \cos(\theta)\,, \\
    \dot{y} =\;& v\sin(\theta)\,, \\
    \dot{\theta} =\;& \omega\,, 
\end{align}
\end{subequations}
where $(x, y) \in \mathbb{R}^2$ is the position of the unicycle in Cartesian coordinates, $\theta \in \mathbb{R}$ is the heading angle, $v$ is the forward velocity input, and $\omega$ is the angular velocity input. In \cite{todorovski2026nonholonomicrobotparkingfeedback, kim2025nonholonomicrobotparkingfeedback}, it was shown that, by transforming the unicycle into polar coordinates, one can obtain a strict CLF (using the bi-directional backstepping design in \cite[Eq. (37), $k_1=k_2=k_3=1$]{kim2025nonholonomicrobotparkingfeedback}) in the form:
\begin{align}
    V(\rho,\delta,\vartheta)
    =
    \rho^2+\delta^2+
    \left(\vartheta+\frac{1}{2}\arctan(2\delta)\right)^2.
\end{align}
where the dynamics in polar coordinates are
\begin{subequations}
\begin{align}
    \dot\rho &= -v\cos\vartheta, \\
    \dot\delta &= \frac{v}{\rho}\sin\vartheta, \\
    \dot\vartheta &= \frac{v}{\rho}\sin\vartheta - \omega .
\end{align}
\end{subequations}
Let
\begin{align}
z:=\vartheta+\frac12\arctan(2\delta).
\end{align}
Then the Lie derivatives are
\begin{align}
L_{g_1}V(\rho,\delta,\vartheta)
&=
-2\rho^2\cos\vartheta
+2\sin\vartheta\left[
\delta+\left(1+\frac{1}{1+4\delta^2}\right)z
\right],\\
L_{g_2}V(\rho,\delta,\vartheta)
&=
-2z.
\end{align}

Consider the general $\gamma$ parameterized feedback law,
 where the physical inputs are given by 
 \begin{align}
     v^\ast(\rho, \delta, \varphi) &= \rho \operatorname{sgn} (-L_{g_1}V(\rho, \delta, \varphi)) \kappa(|L_{g_1}V(\rho, \delta, \varphi)|)\,,  \\ 
     \omega^\ast(\rho, \delta, \varphi) &=  \operatorname{sgn} (-L_{g_2}V(\rho, \delta, \varphi)) \kappa(|L_{g_2}V(\rho, \delta, \varphi)|)\,.
 \end{align}
The corresponding infinite-horizon cost is
\begin{align}
J^\ast
&=
\int_0^\infty
\Bigg[q(\rho,\delta,\vartheta) + r_1(\rho, \delta, \varphi)\gamma\left(\frac{|v|}{\rho} \right)
+
r_2(\rho, \delta, \varphi) \gamma\!\left(|\omega|
\right)
\Bigg]dt,
\end{align}
where
\begin{align}
q(\rho,\delta,\vartheta)
&= (L_{g_1}V)^2 + (L_{g_2}V)^2\,, \\
r_1(\rho,\delta,\vartheta)
&=
\frac{\bigl|L_{g_1}V(\rho,\delta,\vartheta)\bigr|}
{
(\frac{\ell\gamma}{\rm Id})^{-1}\!\left(
\,\bigl|L_{g_1}V(\rho,\delta,\vartheta)\bigr|
\right)},\\
r_2(\rho,\delta,\vartheta)
&=
\frac{\bigl|L_{g_2}V(\rho,\delta,\vartheta)\bigr|}
{
(\frac{\ell\gamma}{\rm Id})^{-1}\!\left(\,\bigl|L_{g_2}V(\rho,\delta,\vartheta)\bigr|
\right)
}.
\end{align}

Second, as a numerical example, we additional provide the explicit control design for the half-Sontag formula in Corollary \ref{cor:half_sontag} where $\gamma(s) =s^2$ and $\kappa(s) = 2s$ is given explicitly. Hence, we obtain
\begin{align} \label{eq:nominal-1}
v_{\rm s/2}^\ast(\rho,\delta,\vartheta)
&=
-\rho \,L_{g_1}V(\rho,\delta,\vartheta) \nonumber \\
&= 2\rho^3\cos\vartheta
-2\rho\sin\vartheta\left[
\delta+\left(1+\frac{1}{1+4\delta^2}\right)z
\right],\\
\omega_{\rm s/2}^\ast(\rho,\delta,\vartheta)
&=
-L_{g_2}V(\rho,\delta,\vartheta)
=
2z. \label{eq:nominal-2}
\end{align}
The half-Sontag law is inverse optimal for the quadratic cost
\begin{align}
J_{\rm S/2}
&=
\int_0^\infty
\left[
q(\rho, \delta, \varphi)
+\left(\frac{v}{\rho}\right)^2
+\omega^2
\right]dt\,. 
\end{align}

\begin{figure*}[!t]
    \centering
    \includegraphics[width=\textwidth]{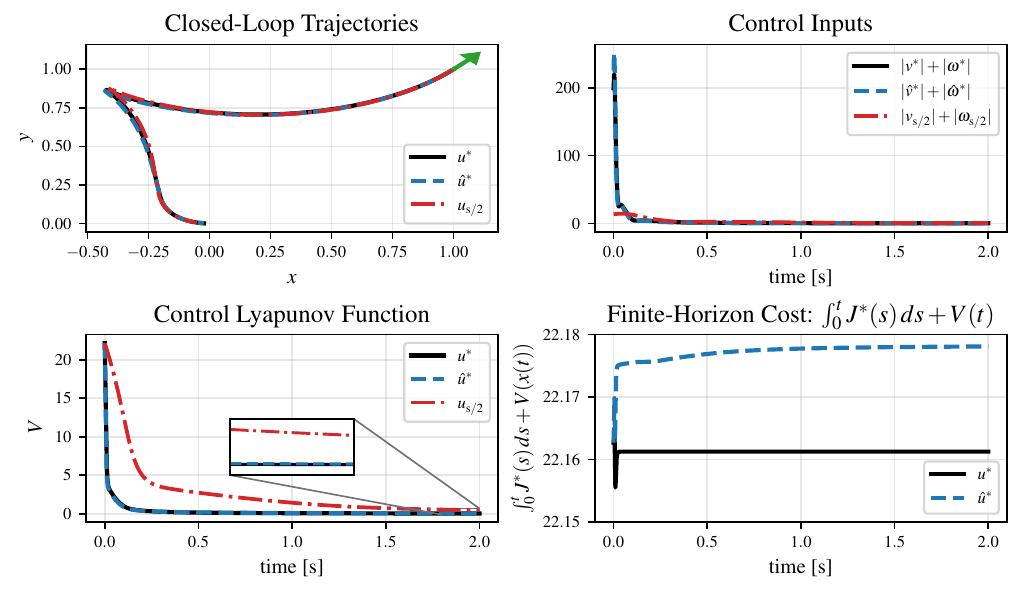}
    \caption{Closed-loop trajectories of the unicycle system \eqref{eq:unicycle-system} from the initial condition $(x,y,\theta)=(1,1,\pi/4)$ under three control laws: the inverse-optimal controller of Example \ref{ex1} with $\alpha=10$ and $\kappa$ computed to numerical precision at each time step, its neural-operator approximation, and the nominal controller \eqref{eq:nominal-1}, \eqref{eq:nominal-2}. The green arrow indicates the initial heading.}
    \label{fig:main-numerical}
\end{figure*}

\subsection{Learning a neural operator approximation of the expander mapping}
As mentioned in Section \ref{sec:lipschitz-approx}, computing $\kappa$ requires inverting $\Theta$ and hence will require a root-finding step at every point of the closed-loop simulation making deployment expensive. Hence, in the section, we learn an approximation of the expander map $\gamma \mapsto \kappa$, once offline, and readily deploy in online capitalizing on the fast inference of neural networks. To do so, consider the $\gamma $ family in Example \ref{ex1} and generate $200$ input output pairs $(\gamma(\cdot), \kappa(\cdot))$ where $\alpha$ is generated log-uniformly from $[10^{-4}, 10^4]$. The pairs for $\kappa$ are generated numerically on a grid of $1024$ points evenly spaced from $[0, 25]$ via the bisection method to machine precision. 

For this study, we deploy the Fourier Neural Operator (FNO) \cite{li2021Fourier} which takes the concatenated input $[\gamma(s),s]\in\mathbb{R}^{1024\times 2}$ to $\hat{\kappa}(s)\in\mathbb{R}^{1024}$. It consists of two layers of $32$ neurons each with $16$ Fourier modes. 
After training, the model achieves relative mean squared error of $0.15$ on the validation set. At inference, notice that $s$ may not be in our numerical grid of $1024$ and hence, between $s$ values, we use a linear interpolation based on the nearest output point of the machine learning model. 

\subsection{Illustrative results}
In Figure \ref{fig:main-numerical}, we present closed-loop trajectories with three controllers:
\begin{enumerate}
    \item \emph{Exact inverse-optimal} ($u^\ast$): the half-direct-optimal law
    with $\kappa$ evaluated by bisection at every time step.
    \item \emph{FNO approximation} ($\hat{u}^\ast$): the same law with $\kappa$
    replaced by a single FNO forward pass with the aforementioned linear interpolation.
    \item \emph{Nominal} ($u_{\rm s/2}$): the quadratic-cost-optimal law
    
    \eqref{eq:nominal-1}--\eqref{eq:nominal-2}.
\end{enumerate}
All three controllers drive the unicycle to the origin. Moreover, the exact and FNO trajectories are visually indistinguishable, confirming that the
neural-operator approximation faithfully replicates the expander mapping in closed loop. As expected, 
the CLF $V$ decreases monotonically to zero under all three laws, verifying stability. However, notice the inverse-optimal design yields a more aggressive transient requiring both a large control cost, but faster CLF convergence. Moreover, this is reflected in the finite-horizon $J^\ast$ as one can see that the exact inverse-optimal design is indeed a better minimizer of the induced cost function rather then the learned $\hat{u}^\ast$. 

\section{Conclusions}

We have introduced a functional parametrization of half-direct-optimal stabilizing feedback laws through the control cost, 
for nonlinear control-affine systems. The parametrization is indexed by a user-chosen function
$\gamma$ serving as the control penalty in an infinite-horizon cost functional. From each
admissible $\gamma$, a nonlinear expander $\kappa$ is constructed through a three-step
procedure, and the resulting feedback law is both globally asymptotically stabilizing and
inverse optimal by design. The cost precedes the controller --- in contrast to the reverse, where, with controller-space parametrizations,  inverse optimality is recovered a posteriori.
The classical Sontag formula is recovered as a special case, and the Freeman-Kokotovic
minimum-norm formula serves as the nominal law that the expander reshapes.

The map $\gamma \mapsto \kappa$ is Lipschitz on compact families of cost functions, and this
regularity guarantees uniform approximation of the entire operator by a neural network.
Under the learned approximation, semiglobal practical asymptotic stability is preserved, and
the induced suboptimality is second order in the approximation error --- a quadratic penalty
for a linear approximation tolerance.

Taken together, these results reframe a classical existence-theoretic construction as a
practical design methodology: the practitioner selects a family of cost functions encoding
different preferences over control aggressiveness, and obtains a corresponding family of
certified stabilizing laws from which to choose on the basis of closed-loop performance.
Numerical results illustrating the learning and quantification of the neural expander operator,
and its deployment in unicycle stabilization simulations, are provided in Section~\ref{sec:num}.

\section*{Dedication}

This paper is dedicated to Eduardo Sontag on the occasion of his 75th birthday in 2026.

\bibliography{references}
\bibliographystyle{elsarticle-num.bst}

\end{document}